%% file: main.tex
\documentclass[conference]{IEEEtran} 
\IEEEoverridecommandlockouts
\usepackage{graphicx}
\usepackage{epsfig}
\usepackage{subfigure}
\usepackage{tabularx}
\usepackage{amssymb}
\usepackage{latexsym}
\usepackage{amsmath}
\DeclareMathOperator*{\argmax}{arg\,max}
\DeclareMathOperator*{\argmin}{arg\,min}

\usepackage{algorithm}
\usepackage[noend]{algpseudocode}
\usepackage{color}
\usepackage{relsize}
\usepackage{mathtools}
\usepackage{amsthm}
\usepackage[,skip=3pt]{caption}
\usepackage{cite}

\newtheorem{theorem}{Theorem}[section]

\theoremstyle{definition}
\begin{document}


\title{AoI-minimizing Scheduling in UAV-relayed IoT Networks}

%
%
%


\author{Biplav Choudhury,\thanks{This research was supported in part by ONR under MURI Grant N00014-19-1-2621, Virginia Tech Institute for Critical Technology and Applied Science, and by Virginia Commonwealth Cyber Initiative (CCI). CCI is an investment in the advancement of cyber R\&D, innovation, and workforce development (www.cyberinitiative.org). Vijay K. Shah is with Cybersecurity Engineering (CYSE) Department at George Mason University, and rest of the authors are with the ECE department at Virginia Tech, Blacksburg, VA. 
(emails:\{biplavc, aidin, reedjh, thou\}@vt.edu and vshah22@gmu.edu).

This paper has been accepted for publication in IEEE MASS 2021.
This is a preprint version of the accepted paper.}

        Vijay K. Shah,
        Aidin Ferdowsi,
        Jeffrey H. Reed,
        and 
        Y. Thomas Hou

}


\maketitle

\begin{abstract}

Due to flexibility, autonomy and low operational cost, unmanned aerial vehicles (UAVs), as fixed aerial base stations, are increasingly being used as \textit{relays} to collect time-sensitive information (i.e., status updates) from IoT devices and deliver it to the nearby terrestrial base station (TBS), where the information gets processed. In order to ensure timely delivery of information to the TBS (from all IoT devices), optimal scheduling of time-sensitive information over two hop UAV-relayed IoT networks (i.e., IoT device to the UAV [hop 1], and UAV to the TBS [hop 2]) becomes a critical challenge. To address this, we propose scheduling policies for Age of Information (AoI) minimization in such two-hop UAV-relayed IoT networks. To this end, we present a low-complexity MAF-MAD scheduler, that employs Maximum AoI First (MAF) policy for sampling of IoT devices at UAV (hop 1) and Maximum AoI Difference (MAD) policy for updating sampled packets from UAV to the TBS (hop 2). We show that MAF-MAD is the optimal scheduler under ideal conditions, i.e., error-free channels and generate-at-will traffic generation at IoT devices. On the contrary, for realistic conditions, we propose a Deep-Q-Networks (DQN) based scheduler. Our simulation results show that DQN-based scheduler outperforms MAF-MAD scheduler and three other baseline schedulers, i.e., Maximal AoI First (MAF), Round Robin (RR) and Random, employed at both hops under general conditions when the network is small (with 10's of IoT devices). However, it does not scale well with network size whereas MAF-MAD outperforms all other schedulers under all considered scenarios for larger networks.

\end{abstract}


\IEEEpeerreviewmaketitle

\input{1-introduction}

\input{2-related-work}

\input{3-system-model}

\input{4-schedulers}

\input{6-results}

\input{7-conclusion}



\ifCLASSOPTIONcaptionsoff
  \newpage
\fi

\bibliographystyle{IEEEtran}
\bibliography{mybib}

\end{document}

%% file: 1-introduction.tex
\section{Introduction}
\label{sec:intro}

Unmanned Aerial Vehicles (UAVs), as fixed aerial base stations, are being increasingly used to collect time-sensitive information (or status updates) from Internet of Things (IoT) devices and forward them to the nearest terrestrial base station (TBS). We refer to such communication networks as \textit{UAV-relayed IoT networks}. UAV-relayed IoT networks are a promising solution in several scenarios -- (i) extending wireless services to IoT devices in remote smart agricultural farms where there is no direct wireless coverage of a TBS \cite{8965135}, (ii) offloading data processing task of a UAV to the TBS equipped with an edge server \cite{7906542} etc. An example is the smart city plans in Las Vegas where multiple UAVs are being deployed as ``advanced contextual and situational awareness modules" to monitor and assess the data captured by IoT devices \cite{HowLasVe14:online}.

A critical requirement in UAV-relayed IoT networks (and general IoT networks)
is that the information received at the TBS from the IoT devices is timely or \textit{fresh} \cite{7415972}. To measure the level of information freshness, the concept of ``Age of Information'' (AoI) was conceived~\cite{5984917}. AoI metric has since attracted significant interest from the research community and is under active investigation. See a survey on AoI~\cite{9380899} and an online bibliography~\cite{webhomea5:online}. AoI is defined as the time elapsed since the most recent information was generated. AoI is fundamentally different from traditional metrics such as, latency, at lower network layers (transport, network or link layer), which only focuses on the elapsed time for delivering information between two end nodes in a network. 


There has been active research on designing scheduling policies to minimize AoI in UAV-assisted IoT networks~\cite{9151993}, general IoT networks~\cite{9305697}, vehicular networks, and other time-sensitive control applications~\cite{choudhury2020joint}. However, existing research on AoI scheduling (including UAV-assisted IoT networks) has been largely limited to single-hop wireless networks, which is not applicable to our UAV-relayed IoT network scenario. This is because our considered UAV-relayed IoT network involves two hops -- (i) [hop 1] - IoT device to UAV, and (ii) [hop 2] - UAV to the TBS. Scheduling for such two-hop network differ from single-hop networks as it involves two time-ordered steps -- (i) \textit{sampling} of IoT devices by the UAVs, followed by the (ii) \textit{updating} of these sampled information packets (from UAV) to the TBS. Therefore, designing scheduling policies for such two-hop UAV-relayed IoT networks becomes a non-trivial problem and is the focus of this work.
Further, existing AoI research make a number of assumptions -- (i)  \textit{non-lossy channel conditions} - In reality, wireless channels are lossy and can change rapidly (e.g., for each time slot) and must be taken into consideration in real-time scheduling (ii) \textit{generate-at-will} traffic generation model at IoT devices -- IoT devices usually follow periodic traffic generation model, where an IoT device generates sampling updates at certain fixed periodic time intervals, which may be unknown apriori~\cite{aziz2019effective}.

In this article, we investigate designing optimal scheduling policies for two-hop UAV-relayed IoT networks under real-world conditions. Specifically, we consider the following: (i) we consider lossy channel conditions for both hops i.e., (IoT devices to UAV) and (UAV to the TBS), (ii) we allow each IoT device to generate packets at periodic traffic generation models, which may be unknown to the scheduler.

The main contributions of this paper are the following.

\begin{itemize}
    \item This paper investigates designing AoI scheduling policies for two-hop UAV-relayed IoT networks, which involves two time-ordered steps -- (i) \textit{sampling} of information packets from IoT devices to the UAV relay, and (ii) \textit{updating} of sampled information packets (corresponding to each IoT device) from UAV relay to the TBS. 
    \item We first present a Maximal AoI First - Maximal AoI Difference (MAF-MAD) scheduling policy. Specifically, MAF-MAD scheduler prioritizes sampling of an IoT device to the UAV relay which has highest AoI at the UAV relay, and similarly, updates the sampled information packet (at UAV relay) of the device with highest AoI difference between the UAV and the TBS. We show that MAF-MAD is the optimal scheduler under non-lossy channel conditions and generate-at-will traffic generation models at the IoT devices. 
    \item Next, we propose an AoI minimizing scheduling policy, based on Deep Q Network (DQN), which we call, DQN based scheduler. We show that DQN based scheduler significantly outperforms MAF-MAD scheduler when the channels are lossy and IoT devices follow periodic traffic generation models, esp. in case of small networks (with 10's IoT devices). This is because DQN is able to learn the characteristics of the network, i.e., channel fluctuations and traffic generation models. However, it does not scale well with network size mainly due to large action space (as we shall see later in the paper).
    \item Our simulation results show that both MAF-MAD and DQN based schedulers outperform all considered baseline schedulers, namely, Maximal AoI First (MAF), Round Robin (RR) and Random schedulers employed at both hops under all considered simulation scenarios. 
\end{itemize}

The rest of the paper is arranged as follows: Sec. \ref{sec:related} discusses the related works and Sec. \ref{sec:system} provides the system model for our work. In Sec. \ref{sec:schedulers}, we explain the different schedulers and their performance in minimizing AoI is presented in Sec. \ref{sec:results}. The paper is concluded in Sec. \ref{sec:conclusion}.

%% file: 2-related-work.tex
\section{Related Work}
\label{sec:related}


As the availability and ease of deployment of UAVs continue to improve, its role in IoT networks under the context of information freshness has seen great interest in recent years. In \cite{abd2019deep}, the authors study a scenario where a single battery constrained UAV monitors multiple ground nodes. A DL based approach was used to develop a joint trajectory and scheduling design that minimizes AoI at the UAV. The work in \cite{ahani2020age} considers a similar setting where a UAV has to visit multiple sensor nodes under battery constraints. Similarly, a framework for finding the trajectory that minimizes the average AoI and maximum AoI is presented in \cite{liu2018age}. A cellular Internet of UAVs is considered in \cite{9014214} where a distributed trajectory selection framework is proposed to minimize AoI. 
Note all these works (along with \cite{zhou2019deep, 9005434, hu20201, 9162896}) focus mainly on determining the optimal trajectory of the UAVs with the scheduling of packets being only over one hop from IoT devices to the UAV.

Since our UAV-relayed IoT network employ UAV as relays between IoT devices and the BS, we also discuss the related works on AoI for relay networks. Authors in \cite{buyukates2018age} consider a scenario with a single node transmitting to multiple nodes via relays. The source waits for an acknowledgment before transmitting the next sample, and the best waiting time is obtained in closed form. Reference \cite{arafa2019timely} proposes optimal AoI minimizing online and offline policies where a single energy constrained source transmits its information to a single receiver via an energy constrained relay. In \cite{moradian2020age}, a discrete-time stochastic hybrid system analysis is done, first under one relay and a direct link, and then with two relays and no direct link between a single transmitter and a single receiver. Our work is closest to \cite{song2020optimal} where a relay is used to sample and forward the information of multiple IoT devices to multiple destinations. However, it considers a generate-at-will traffic model for traffic generation at the IoT devices with only a single UAV acting as the relay. 



%% file: 3-system-model.tex
\section{System Model and Problem Formulation}
\label{sec:system}

\begin{figure} [htb]
    \centering
    \vspace{-0.1 in}
    \includegraphics[scale=0.30, trim={5cm 0cm 4cm 1.2cm},clip, angle=0]{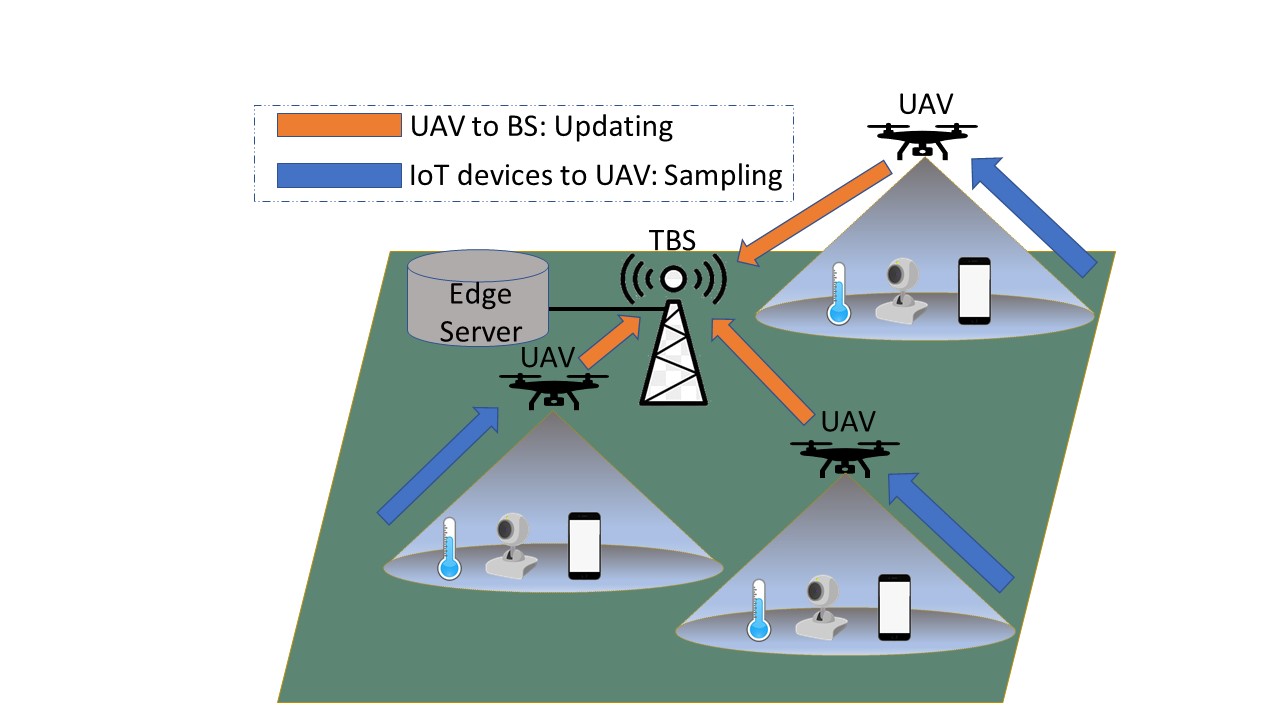}
    \vspace{0in}
    \caption{UAV-relayed IoT networks system model}
    \label{fig:scenario}
\end{figure}



\subsection{Network Model} 
As shown in Figure \ref{fig:scenario}, we consider a UAV-relayed IoT network with a terrestrial base station (TBS), $M$ IoT devices and $N$ UAVs. The TBS is placed at a (random) fixed 2D location in a certain geographical region. The TBS is equipped with edge server, where the collected time-sensitive information data are processed. Similarly, $M$ IoT devices are distributed randomly in the considered region. Since IoT devices are usually energy-constrained and have limited wireless range (most IoT devices have range less than 100 meters~\cite{HowIoTSh99:online}), the IoT devices are outside the wireless coverage of the nearest TBS, such as, in case of rural agricultural smart farms, where the IoT devices are far away from the nearest TBS. Therefore $N$ UAVs are deployed as fixed aerial base stations (ABS) in the geographical region to provide end-to-end wireless coverage between the nearest TBS and all $M$ IoT devices. Each UAV $n \in N$ is deployed in the region of interest such that  

(i) UAV $n$ lies within the wireless coverage of the TBS and, provides \textit{backhaul links} between itself and TBS and 

(ii) UAV $n$ provides wireless coverage as \textit{access links} to $m_n$ subset of IoT devices, where $m_n \subseteq M$. 

Even though a certain IoT device may fall within the wireless coverage of two or more UAVs, we assume that each IoT device is associated to a unique UAV based on some association policy (e.g., max power association policy~\cite{8723312}), and thus, the subset of IoT devices, $m_n$, is unique to each UAV, and $M = \cup_{n \in N} m_n$. 
We assume that there is neither intra-cell interference (i.e., UAV and it's associated IoT devices) nor inter-cell interference (i.e., UAV to UAV). This is realistically possible by employing non-overlapping orthogonal channels for both access links \cite{9419071} and backhaul links \cite{8932190}. 

Let $K$ denote the number of non-overlapping orthogonal channels for 
communication between the TBS and all $N$ UAVs. Similarly, let $L_n$ denote the total number of channels available between a certain UAV $n$ and it's associated IoT devices. For easiness of the presentation, we consider a fixed $L = L_n, (\forall n \in N$) channels for access links between a certain UAV and it's associated IoT devices. However, our model can be easily extended to unique $L_n$ channels for access links between a UAV and its associated IoT devices. 

Next we discuss different traffic generation models at each IoT device and channel conditions considered in our model. Assume $T$ is the total observation duration and it is divided into equal-length time slots denoted by $t$. 

\textbf{Traffic generation at the IoT devices.} We consider two different types of traffic generation models at each IoT device $m$ -- (i) \textit{generate-at-will}: $m$ generates new packet at each time slot $t$~\cite{abd2019deep}, and (ii) \textit{periodic packet generation}: Here, $m$ generates information packets at a fixed time interval denoted by $p_m$. It means, $m$ generates a new packet at time slots $p_m t, 2p_m t, \dots, \lfloor \frac{T}{p_m} \rfloor t$. Note that though $p_m$ is a fixed value for a certain IoT device $m$, it may  be unique to each IoT device, i.e., $p_m \neq p_{m'}$, where $m, m' \in M$. The scheduler doesn't have any prior information on the periodicity of traffic generation at a certain IoT device.

\textbf{Channel Conditions:} Because of the fluctuations in the wireless channel, some packets in both the two hops are lost in transmission. 
It is referred to as link outage and occurs whenever the transmission rate exceeds the channel capacity \cite{tse2005fundamentals}.
Similar to \cite{song2020optimal}, we capture the link outage by considering that the link between device $m$ and the associated UAV $n_m $ has a non-zero probability of dropping the transmitted packets, which we refer to as, \textit{Sample Loss Probability} and is denoted by $l_{s,m} \in (0,1)$. Similarly, $l_{u,m} \in (0,1)$ denotes the \textit{Update Loss Probability}, which refers to the packet loss probability between UAV $n_m$ and the TBS.

\textbf{Scheduling:} At each slot $t$, our considered UAV-relayed IoT networks involves two simultaneous scheduling steps -- (i) \textit{sampling} of IoT devices by the UAVs and (ii) \textit{updating} of sampled information packets from a certain UAV to the TBS. The set of IoT devices sampled by UAV $n$ is denoted by $S_n(t)$ and total devices sampled by all the UAVs is $\mathcal{S}(t)=\{\mathcal{S}_1(t),\mathcal{S}_2(t), \dots, \mathcal{S}_N(t)\}$. Similarly, updated devices are denoted as $\mathcal{U}(t)$. When a device $m$ is sampled, it transmits its most recent packet to the serving UAV $n_m$ and it will replace any of its older packet at $n_m$ \cite{8885156}. Then when it is updated, the packet is transmitted from UAV $n_m$ to the TBS. At any slot, each channel can support at most 1 packet. 

\subsection{Age of Information (AoI) at the UAV and TBS} We employ Age of Information (AoI) to measure the freshness of information. In particular, AoI is defined as \textit{time elapsed since the generation time of the most-recent packet received 
(by the UAV/TBS in our case).} 
The AoI of IoT device $m$ at UAV $n_m$ is denoted as 

\begin{align}
    AoI_m^{UAV} (t) = t - \tau_{s,m} 
    \label{eqn:AoI_evolution_UAV}
\end{align}

\noindent where $t$ is the current slot and $\tau_{s,m}$ is the generation time of the \textit{most recently sampled} packet of device $m$ that was received successfully at UAV $n_m$. 
Similarly, the AoI of device $m$ at the TBS is 
\begin{align}
    AoI_m^{TBS} (t)  = t - \tau_{u,m}
    \label{eqn:AoI_evolution_TBS}
\end{align}

\noindent where $\tau_{u,m}$ is the generation time of the \textit{most recently updated} packet of the $m$th IoT device that was received successfully at the TBS. Based on Eqn. \eqref{eqn:AoI_evolution_UAV} and \eqref{eqn:AoI_evolution_TBS}, AoI increases linearly in time slots of no reception and decreases at reception instants.

The evolution of AoI with time at the UAVs and TBS is described next. At $t$, if device $m$ was selected for sampling, i.e., $m \in \mathcal{S}(t)$, its AoI at the UAV changes as 

\begin{equation} \label{eqn:transition_AoI_UAV_sample}
    AoI_{m}^{UAV}(t+1) =
    \begin{cases}
    t+1-\tau_{s,m} & \text{with prob.} \hspace{.05in} 1-l_{s,m}  \\
    AoI_{m}^{UAV}(t)+1 & \text{with prob.} \hspace{0.05in} l_{s,m} 
    \end{cases}
\end{equation}

Else if the device $m$ was not sampled, i.e., $m \notin \mathcal{S}(t)$, its AoI at the UAV changes as 

\begin{equation} \label{eqn:transition_AoI_UAV_no_sample}
    AoI_{m}^{UAV}(t+1) = AoI_{m}^{UAV}(t)+1
\end{equation}

Similarly AoI for IoT device $m$ at the TBS when selected for update such that $m \in \mathcal{U}(t)$ and when not selected for update such that $m \notin \mathcal{U}(t)$ are shown in Eqn. \eqref{eqn:transition_AoI_TBS_update} and  \eqref{eqn:transition_AoI_TBS_no_update}.

\begin{equation} \label{eqn:transition_AoI_TBS_update}
    AoI_{m}^{TBS}(t+1) =
    \begin{cases}
     AoI_{m}^{UAV}(t)+1 & \text{with prob.} \hspace{.05in} 1-l_{u,m}  \\
    AoI_{m}^{TBS}(t)+1 & \text{with prob.} \hspace{0.05in} l_{u,m} 
    \end{cases}
\end{equation}

\begin{equation} \label{eqn:transition_AoI_TBS_no_update}
    AoI_{m}^{TBS}(t+1) = AoI_{m}^{TBS}(t)+1
\end{equation}

Thus a device's AoI increases at the UAVs and the TBS if- 
\begin{itemize}
    \item it was sampled/updated but the packet was lost in transmission 
    \item it was not sampled/updated.
\end{itemize}

As the information is relayed through the UAVs, $AoI_{m}^{UAV}$ directly impacts $AoI_{m}^{TBS}$. Note that the communication between a device and the TBS has a delay of a single slot, due to which information sampled by an UAV at a certain slot cannot be updated to the TBS in the same slot \cite{song2020optimal}. 

Finally, the average AoI of all IoT devices at the UAVs and the TBS during the observation interval $T$ is calculated as Eqn. \eqref{eqn:UAV_AoI_avg} and \eqref{eqn:TBS_AoI_avg} respectively

\begin{align}
    AoI^{UAV}(T) = \frac{1}{M} \sum_{t=1}^{T} \sum_{m=1}^{M} AoI_{m}^{UAV} (t)
    \label{eqn:UAV_AoI_avg}
\end{align}


\begin{align}
    AoI^{TBS}(T) = \frac{1}{M} \sum_{t=1}^{T} \sum_{m=1}^{M} AoI_{m}^{TBS} (t)
    \label{eqn:TBS_AoI_avg}
\end{align}

\textit{Problem formulation:} Our objective to design a scheduler that ensures minimum AoI corresponding to all IoT devices at the TBS. Given $S_n(t)$ and $\mathcal{U}(t)$ respectively denote the set of the IoT devices sampled by each UAV $n \in N$ and updated to the TBS at a certain time slot $t \in T$, we formulate the AoI-aware scheduling problem as the minimization of the average AoI of all IoT devices at the TBS (Eqn. \eqref{eqn:TBS_AoI_avg}) subject to the limited channel constraints, as follows.

 \begin{align}
 	\nonumber & \min_{} AoI^{TBS}(T) \\
 	\text{s.t. } & \left| \mathcal{S}_{n}^{}(t) \right| \le L, \hspace{5pt} \forall n \in N \hspace{5pt} \text{and} \hspace{5pt} t=1,2...T \label{constraint1}
 	\\
 	& \left| \mathcal{U}^{}(t) \right| \le K, \hspace{5pt} t=1,2...T \label{constraint2}
 \end{align}
 
\noindent where constraint \eqref{constraint1} and \eqref{constraint2} refers to the limited number of channels for the sampling and updating respectively. The notations used in this paper are summarized in Table \ref{tab:notations}.


\begin{table}
 \centering
 \small
 \vspace{-0.18in}
 \caption{Notations}
    \begin{tabular}{|p{5.9cm}|p{2cm}|}
    \hline
    \textbf{Meaning }  & \textbf{Symbol} \\ \hline
    Number of IoT devices  & $M$\\ \hline
    Number of UAVs  & $N$\\ \hline
    IoT devices associated to UAV $n$  & $m_n$ \\ \hline
    UAV providing coverage to device $m$ & $n_m$ \\ \hline
    Packet generation periodicity of device $m$ & $p_m$  \\ \hline
    AoI of device $m$ at UAV & $AoI_m^{UAV}$ \\ \hline
    AoI of device $m$ at TBS & $AoI_m^{TBS}$ \\ \hline
    Packet loss between device $m$ and UAV $n_m$ & $l_{s,m}$ \\ \hline
    Packet loss between device $m$ and TBS & $l_{u,m}$ \\ \hline
    Devices sampled at $t$ by UAV $n$  & $\mathcal{S}_n(t) $ \\ \hline
    Devices sampled at $t$=$\{\mathcal{S}_1(t), \mathcal{S}_2(t),..\mathcal{S}_N(t)\} $  & $\mathcal{S}(t) $ \\ \hline
    Devices updated at $t$  & $\mathcal{U}(t) $  \\ \hline
    Channels at each UAV & $L$ \\ \hline
    Channels at the TBS & $K$ \\ \hline
    
 \end{tabular}
 \label{tab:notations}
\end{table}

%% file: 4-schedulers.tex
\section{Scheduling Policies}
\label{sec:schedulers}

We first propose a simple low-complexity centralized , \textit{Maximal AoI First-Maximal AoI Difference} (MAF-MAD) scheduling algorithm to solve the above problem. MAF-MAD runs at the TBS and its underlying idea is to sample users with maximum AoI at the serving UAV, and update IoT devices with maximum AoI difference between the UAV and the TBS. We show that MAF-MAD is an optimal scheduler for UAV-relayed IoT networks under ideal conditions where (i) the channels are non-lossy and (ii) all the IoT devices follow generate-at-will traffic generation models. However, under general conditions (i.e., lossy channels and periodic traffic generation models), the scheduling decisions need to incorporate packet generation instants and channel non-idealities, which is not possible in case of MAF-MAD scheduler. Thus, we propose a centralized learning algorithm based on Deep Q networks (DQN), which we call \textit{DQN-based scheduler}. DQN based scheduler also runs at the TBS and learns the unique network characteristics, and thus, promises to perform well for UAV-relayed IoT networks under realistic conditions.

\subsection{MAF-MAD Scheduler}
\label{mad_policy}



The algorithmic details of Maximal AoI First - Maximal AoI Difference (MAF-MAD) scheduler are presented in  Algorithm \ref{algorithm - MAD}. At any time $t \in T$, the inputs to the this algorithm are AoI of each IoT device $m \in M$ 
at its serving UAV $AoI_m^{UAV}(t)$ and at the TBS $AoI^{TBS}_m(t)$. For the sampling step in line \ref{MAD-sample}, each UAV $n$ samples $L$ IoT devices with the maximal AoI, $AoI_m^{UAV}(t)$, first (MAF) from the $m_n$ IoT devices associated with it. Note that ideally we should consider AoI difference between the UAV and the IoT device, but this is not feasible as the scheduler doesn't know the traffic generation instants of the IoT devices. However when devices generate packets using generate-at-will traffic model, any device which is selected to be sampled will generate a new packet at that slot due to which $AoI_m^{UAV}(t)$ also becomes the AoI difference between the UAV and the IoT device.

Then for the updating step, devices with \textit{maximal AoI difference} (MAD) between the UAV and TBS are prioritized. This AoI difference is calculated as 

\begin{align}
    AoI_m^{diff} (t) = AoI_{m}^{TBS}(t) - AoI_{m}^{UAV}(t) 
    \label{eqn:AoI_users_diff_update}
\end{align}

\noindent where $AoI_{m}^{TBS}(t)\geq AoI_{m}^{UAV}(t)$. The TBS updates $K$ users out of the total $M$ IoT devices with the highest $AoI_m^{diff} (t)$, as shown in line \ref{MAD-update}.




\begin{algorithm}  
\small
	\textbf{Input:} $AoI_m^{UAV}(t)$ and $AoI_m^{TBS}(t)$ for $m \in M$.
	
	\textbf{Output:} $\mathcal{S}(t), \mathcal{U}(t)$.
	
	\begin{algorithmic} [1]
	\State $\mathcal{S}(t) = \phi$
	\For {UAV $n = 1,2,..N$}:
    \State $\mathcal{S}_n(t) = \argmax_{S \subseteq m_n, |\mathcal{S}_n(t)| \le L} \{AoI_{m}^{UAV}(t)\}_{m=1}^{m_n}$ \label{MAD-sample}
    \State $ \mathcal{S}(t) = \mathcal{S}(t) \cup \mathcal{S}_n(t) $
    \EndFor
    \State $\mathcal{U}(t) = \argmax_{U \subseteq M, |\mathcal{U}(t)| \le K} \{AoI_{m}^{diff}(t)\}_{m=1}^{M} $ \label{MAD-update}
	\end{algorithmic}  
	\caption{MAF-MAD scheduler}
	\label{algorithm - MAD}
\end{algorithm}

\begin{theorem}
Under ideal conditions, i.e., no lossy channels and generate-at-will sampling at IoT devices, MAF-MAD is the optimal scheduling policy for minimizing average AoI at the TBS for status updates generated at each IoT device.
\end{theorem}


\begin{proof}
As also presented in \cite{song2020optimal}, to minimize the AoI at a destination ($AoI^{TBS}(T)$ in our case), at each slot $t=1,2,\dots,T$ the optimal scheduler needs to 

\begin{itemize}
    \item sample devices whose packet reception at the relay (UAV in our case) will lead to the maximum reduction of average AoI at the relay, 
    \item update devices whose packet reception at the TBS will lead to the maximum reduction of average AoI at TBS.
\end{itemize}

We first show that sampling as per MAF-MAD offers the maximum reduction in average AoI at all the UAVs. 

\textbf{AoI reduction at the UAVs:} 
Consider a UAV $n$ 
providing coverage to $|m_n|$ (out of $M$) IoT devices. At slot $t$, the average AoI of these IoT devices at UAV $n$ is 

\begin{align}
    AoI_n^{UAV}(t) = \frac{1}{M} \sum_{i=1}^{m_n} AoI^{UAV}_{i}(t) 
\end{align}

\noindent and from the result cited above, the optimal scheduler will sample those devices which will lead to maximum reduction of average AoI at $n$ from $AoI_n^{UAV}(t)$ to $AoI_n^{UAV}(t+1)$. 

For simplicity, assume $L$=$K$=$1$. Therefore the scheduler needs to select $L$=$1$ device to sample. Let device $j \in m_n$ have the largest AoI at $n$ such that $ AoI_j^{UAV}(t) = \max({AoI_1^{UAV}(t),AoI_2^{UAV}(t), \dots AoI_{m_n}^{UAV}(t)})$. From Algorithm \ref{algorithm - MAD} step \ref{MAD-sample}, at $t$ the MAF-MAD scheduler will sample the $j$th device. Then the reduction in the average AoI at $n$ is

\begin{align}
    \nonumber \Delta AoI_n^{UAV}(t+1) &= \frac{1}{M} (AoI_n^{UAV}(t) - AoI_n^{UAV}(t+1)) \\
    &=\frac{1}{M} \sum_{i=1}^{m_n} (AoI^{UAV}_{i}(t) - AoI^{UAV}_{i}(t+1) )
    \label{eqn:delta_AoI_UAV}
\end{align}

For the $|m_n|-1$ unsampled devices, 
their $AoI^{UAV}(t+1)$ increases by 1 (Eqn. \eqref{eqn:transition_AoI_UAV_no_sample}). $AoI^{UAV}(t+1)$ for $j$ becomes $t+1-\tau_{s,m}$ as per Eqn. \eqref{eqn:transition_AoI_UAV_sample} under no packet loss. Additionally for generate-at-will, $j$ will generate a packet at $t$, i.e., $\tau_{s,m}=t$, as it was sampled at $t$. Therefore Eqn. \eqref{eqn:delta_AoI_UAV} becomes

\begin{align}
    \nonumber \Delta AoI_n^{UAV}(t+1) &= \frac{1}{M} (\sum_{i=1, i\neq j}^{|m_n|}(AoI_{i}^{UAV}(t)-AoI_{i}^{UAV}(t+1))\\ & \nonumber \hspace{2mm}  + (AoI_{j}^{UAV}(t)- AoI_j^{UAV}(t+1))) \\
    \nonumber &=\frac{1}{M}( -(|m_n|-1)+AoI_j^{UAV}(t)-\\& \nonumber  \hspace{4mm} (t+1-t)) \\
    &= \frac{1}{M}(AoI_j^{UAV}(t) - |m_n|)
    \label{eqn:AoI_reduction_optimal_UAV}
\end{align}

Now for any other scheduler that selects device $j'$ to sample where $j' \neq j$ and hence $ AoI_j^{UAV}(t) > AoI_{j'}^{UAV}(t) $, the reduction in AoI at $n$ will be given by 

\begin{align}
    \Delta^{'} AoI_n^{UAV}(t+1) &= \frac{1}{M} (AoI_{j'}^{UAV}(t)-|m_n|)
     \label{eqn:AoI_reduction_non_optimal_UAV}
\end{align}

Now comparing Eqn. \eqref{eqn:AoI_reduction_optimal_UAV} and \eqref{eqn:AoI_reduction_non_optimal_UAV}, $\Delta AoI_n^{UAV}(t+1) > \Delta^{'} AoI_n^{UAV}(t+1)$. Because all the UAVs $n \in N$ operate independently, the same result holds for other UAVs too. Hence the MAF-MAD scheduler offers the largest AoI reduction at all the UAVs in the sampling step. 

\textbf{AoI reduction at the TBS:} In the updating step, the scheduler needs to select $K$=1 device to update. From Algorithm \ref{algorithm - MAD} step \ref{MAD-update}, at $t$ the MAF-MAD scheduler will update the $k$th device such that $AoI_{k}^{diff}(t) = \max(AoI_{1}^{diff}(t), AoI_{2}^{diff}(t), \dots, AoI_{M}^{diff}(t))$. After the update, the reduction in average AoI at the TBS is 

\begin{align}
    \nonumber \Delta AoI^{TBS}(t+1) &=\frac{1}{M} (AoI^{TBS}(t) - AoI^{TBS}(t+1)) \\
    &=\frac{1}{M}( \sum_{i=1}^{M} (AoI^{BS}_{i}(t) - AoI^{TBS}_{i}(t+1)))
    \label{eqn:delta_AoI_TBS}
\end{align}

For $M-1$ devices that were not updated, their $AoI^{TBS}(t+1)$ increases by 1 (Eqn. \eqref{eqn:transition_AoI_TBS_no_update}). $AoI^{TBS}(t+1)$ for the updated device $k$ will change as per Eqn. \eqref{eqn:transition_AoI_TBS_update} when there is no packet loss. Therefore Eqn. \eqref{eqn:delta_AoI_TBS} becomes 

\begin{align}
    \nonumber \Delta AoI^{TBS}(t+1) &= \frac{1}{M} (\sum_{i=1, i\neq k}^{M}(AoI_{i}^{TBS}(t)-AoI_{i}^{BS}(t+1))\\ \nonumber &+ (AoI_{k}^{TBS}(t)-
    AoI_k^{TBS}(t+1))) \\
   \nonumber &= \frac{1}{M} (-(M-1)+AoI_k^{TBS}(t)-\\ \nonumber & AoI_k^{BS}(t+1)) \\
    \nonumber  &= \frac{1}{M}( -M+1+AoI_k^{TBS}(t)-\\ \nonumber & AoI_{k}^{UAV}(t)-1) \\
     &= \frac{1}{M} (AoI_k^{diff}(t)-M)
     \label{eqn:AoI_reduction_optimal_TBS}
\end{align}

Now for any other scheduler that selects device $k'$ to sample where $k' \neq k$ and hence $ AoI_k^{diff}(t) > AoI_{k'}^{diff}(t) $, the reduction in average AoI at the TBS will be given by

\begin{align}
    \Delta^{'} AoI^{TBS}(t+1) = \frac{1}{M}(AoI_{k'}^{diff}(t)-M)
     \label{eqn:AoI_reduction_non_optimal_TBS}
\end{align}

Comparing Eqn. \eqref{eqn:AoI_reduction_optimal_TBS} and \eqref{eqn:AoI_reduction_non_optimal_TBS}, we see $\Delta AoI^{TBS}(t+1)>\Delta^{'}AoI^{TBS}(t+1)$. Hence MAF-MAD also leads to the maximum reduction in average AoI at the TBS. Therefore when there is no packet loss at any of the sampling and updating steps, and the devices generate traffic based on generate-at-will policy, MAF-MAD is the optimal scheduler for minimizing $AoI^{TBS}(T)$.

\end{proof}


\subsection{Deep Q Network (DQN) based scheduler}
\label{dqn_policy}



The requirement of taking the channel non-idealities and traffic generation patterns motivates the use of a reinforcement learning (RL) based scheduler. In RL, an agent is able to take an informed decision on which actions to take after interacting with the environment as the interaction allows it to learn about the environment. A typical RL framework has 
an agent (scheduler in our case) that takes an \textit{action} (devices to sample and update) after observing the environment \textit{state}, following which it gets a \textit{reward} which is like a feedback on the quality of the action taken. We use a Q-learning algorithm \cite{powell2007approximate} to design a scheduler with the objective of AoI minimization.

In Q-learning, a state-action value function $Q^{\pi}(x(t),a(t))$ is defined as the net reward when the system performs action $a(t)$ while starting from state $x(t)$ and then follows policy $\pi$. Q-learning tries to calculate the Q-value function by learning a policy which will maximize the cumulative reward. This is done using the Bellman update rule shown in Eqn. \eqref{eqn:q_function} below.

\begin{multline}
    Q_{t+1}(x(t),a(t))=Q_{t}(x(t),a(t))+ \\ \beta[r(t)+\gamma \max_{a} Q_{t} (x(t+1),a)-Q_{t}(x(t),a(t))]
    \label{eqn:q_function}
\end{multline}

\noindent where $r(t)$ is the reward at slot $t$. $\beta$ and $\gamma$ are the learning rate and the discount factor respectively. Discount factor signifies how important the future rewards are and is set to a value between 0 and 1. As our problem has a terminal state when $t=T$, it is an episodic task and hence $\gamma$ is set to 1 \cite{abd2019deep}.

The agent approximates the Q-function using Eqn. \eqref{eqn:q_function} and then takes the action that maximizes the reward. As the agent will not have the correct estimate of the Q-values corresponding to many state-action pairs as a consequence of not visiting them, it needs to explore in addition to exploiting the known values of state-action pairs. This is known as the exploration-exploitation trade-off. An epsilon ($\epsilon$) greedy approach where the agent takes a random action with the objective of exploring the environment with probability $\epsilon$ and acts greedily with probability $1-\epsilon$ is generally used to ensure the agent doesn't get stuck with sub-optimal actions. 

While Eqn. \eqref{eqn:q_function} can be theoretically used for all scenarios, its iterative nature makes it infeasible to be used for large state spaces as it will suffer from very slow convergence in addition to needing a large memory \cite{abd2019deep}. Moreover this approach cannot be used when there are unobserved states \cite{powell2007approximate} and visiting every state-action pair is not practical. This is where the approximating power of neural networks are put to use as they can extract the important features from the available data points and summarize them in smaller dimensions. Neural networks approximates the Q-value function as $Q(x,a|\theta)$ where $\theta$ are its weights, and this approach for Q-learning is known as Deep Q-Networks (DQN) and is based on a Markov Decision Process (MDP) formulation of the problem \cite{mnih2015human}. 
The goal is to find the optimal values for $\theta$ so that the neural network can approximate the optimal Q-value function as close as possible. 
However the use of a single neural net may cause instability and therefore two neural networks with the same architecture are used: the \textit{current neural network} with weights $\theta$ and a \textit{target neural network} with weights $\theta^{-}$. While the current network is used as a function approximator and its weights are updated iteratively after each slot, the target network computes the target Q-value function and its weights are fixed for a while and updated every $O$ slots \cite{mnih2015human}.

At each time step, the \textit{experience} of the agent is stored in \textit{experience replay} which is also called replay memory and has a size $R$. The experience at time $t$ is denoted as $exp(t) = (x(t),a(t),r(t),x(t+1))$ and provides a summary of the agent’s experience at that time. The replay memory is then used to randomly sample a mini-batch of stored experiences $x(t'),a(t'),r(t'),x(t'+1)$ which will be used to train the neural network. The primary benefit of using random samples from replay memory to learn instead of consecutive samples as the agent encounters them is that consecutive experiences have high correlation which may lead to inefficient learning \cite{8382166}. Therefore at each episode, a mini-batch of $B$ past experiences are fed as an input to the DQN. The loss function is given by 

\begin{equation}
    L(\theta) =   [r(t')+\gamma \max_{a(t')}Q(x(t+1),a(t')|\theta^{-}) - Q(x(t),a(t)|\theta)]^2
    \label{eqn:loss}
\end{equation}

\noindent where $Q(x(t+1),a(t')|\theta^{-})$ is evaluated by the target network and $Q(x(t),a(t)|\theta)$ is evaluated using the current network. The weights are then updated as shown in Eqn. \eqref{eqn:weight_update_current}.

\begin{equation}
    \theta = \theta + \alpha(y(t')-Q(x(t'),a(t')|\theta))\nabla_{\theta}Q(x(t'),a(t')|\theta)
    \label{eqn:weight_update_current}
\end{equation}

\noindent where $y(t')=r(t')+\max_{a}Q(x(t'+1,a|\theta^{-}))$ and $\nabla_{\theta}$ denotes the gradient with respect to $\theta$.

The state, action and reward in the context of our work is explained next. 

\subsubsection{State Space}
State is the scenario encountered by the scheduler. We define our state as

\begin{align}
x(t) = (t, \{{AoI_{m}^{UAV}(t)}\}_{m=1}^{M}, \{{AoI_{m}^{TBS}(t)}\}_{m=1}^{M})
\label{eqn:state_space}
\end{align}

\noindent where $t$ is the ongoing slot. At $t$=1, $AoI_{m}^{UAV}(t)$ and $AoI_{m}^{TBS} (t)$ are initialized to 1. The set of state space is denoted by $\mathcal{X}$.

\subsubsection{Action Space}
As described in Sec. \ref{sec:system}, scheduling involves two simultaneous steps of sampling updating devices at each slot. Therefore the action at each slot $t$ is 

\begin{align}
    a(t) = (\mathcal{S}(t), \mathcal{U}(t))
    \label{eqn:action_space}
\end{align}

\noindent where $\mathcal{S}(t)$=\{$S_1(t), S_2(t), ..S_N(t)$\} is the set of all devices sampled by all the UAVs at slot $t$ such that $|S_n(t)| \le L \hspace{0.05in} \forall \hspace{0.05in} n \in N$. Similarly, $\mathcal{U}(t) $ is the set of devices selected to update their packets to the TBS where $|\mathcal{U}(t)| \le K$. The set of action space is denoted by $\mathcal{A}$.

\subsubsection{Reward}
As the objective is improving the information freshness at the TBS, the reward is given by the negative of the average AoI of all the devices at the TBS at slot $t$ 
\begin{align}
    r(x(t), a(t)) = r(t) = -\sum_{m=1}^{M} AoI_m^{TBS}(t)
    \label{eqn:reward_fn}
\end{align}

Based on action $a(t)$, the environment transitions from state $x(t)$ to a new state $x(t+1)$ according to the state transition probabilities described in Eqn. \eqref{eqn:transition_AoI_UAV_sample}, \eqref{eqn:transition_AoI_UAV_no_sample}, \eqref{eqn:transition_AoI_TBS_update}, \eqref{eqn:transition_AoI_TBS_no_update} while resulting in a reward $r(t)$. Thus it is a finite horizon MDP with finite state and action spaces, which makes it suitable for DQN. 

The DQN based scheduler is shown in Fig. \ref{fig:dqn} and explained in Algorithm \ref{algorithm - dqn}. The hyper-parameters and the network are initialized in lines \ref{dqn_algo:hyp_initialize} - \ref{dqn_algo:network_initialize}. For each step, a random action is taken with probability $\epsilon$ or the action with the highest Q-value is selected with probability 1-$\epsilon$ in lines \ref{dqn_algo:action_select_start} - \ref{dqn_algo:action_select_end}. The selected action is performed and network transitions to the new state in lines \ref{dqn_algo:action_perform}-\ref{dqn_algo:obs_reward}. Line \ref{dqn_algo:store_exp} stores the resulting experience in the replay memory. Then a mini-batch of samples is used to train the network in lines \ref{dqn_algo:train_start}-\ref{dqn_algo:train_end}, after which the weights of the current network are updated in lines \ref{dqn_algo:weight_update_current}. Line \ref{dqn_algo:weight_update_target} updates the target network weight every $O$ slots.

\textbf{Convergence:} Convergence of neural networks are hard to be analytically investigated and is strongly dependent on the set of hyper-parameters used \cite{jomaa2019hyp}. Selection of hyper-parameters is a challenging task and therefore a reasonable set of hyper-parameters is found by trying out different values. More details on the hyper-parameters used in our case are presented in Sec. \ref{sec:results}. Similar to \cite{abd2019deep, 9014214, song2020optimal, zhou2019deep}, we limit our investigation into the convergence to simulations, where the neural network converges under the hyper-parameters used and the results presented for the DQN scheduler are the values obtained after it's convergence.

\begin{figure} [htb]
    \vspace{-0.05 in}
    \centering
    \includegraphics[scale=0.28, trim={1cm 0cm 0cm 0cm},clip, angle=0]{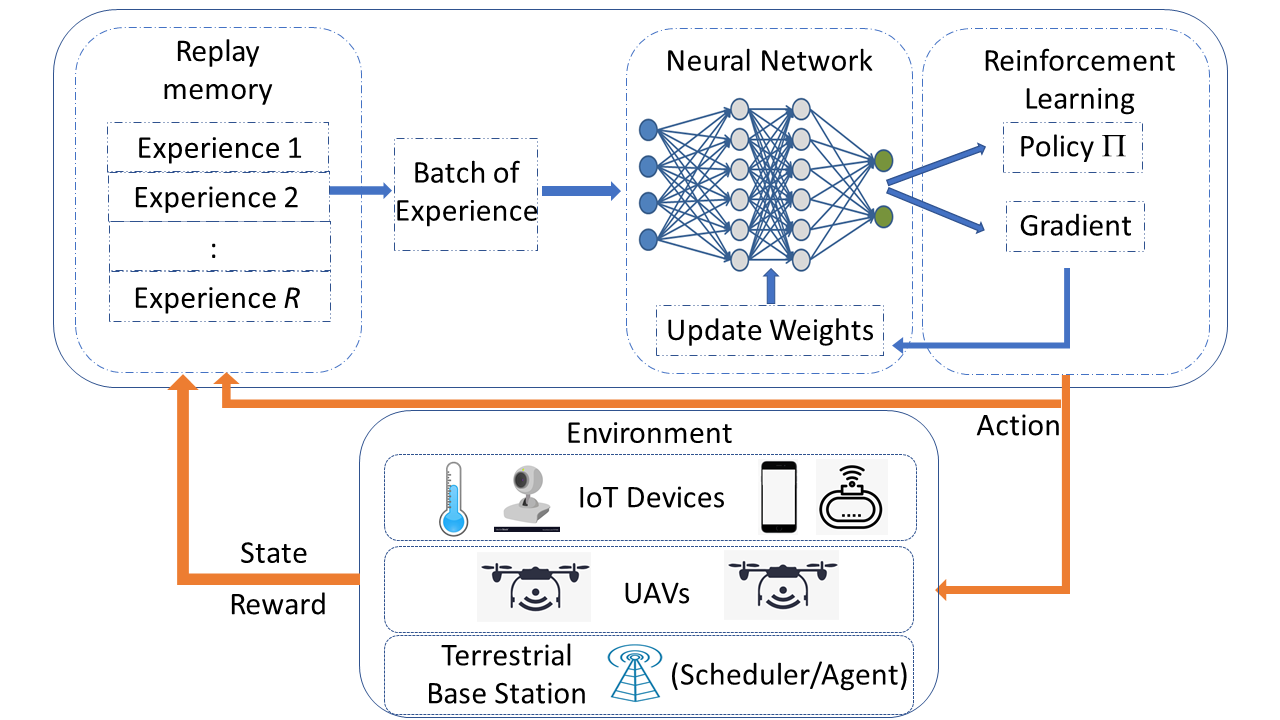}
    \vspace{-0.1in}
    \caption{Deep Q Network}
    \label{fig:dqn}
    \vspace{-0.05 in}
\end{figure}

\begin{algorithm}[t]
	\caption{DQN scheduler}
	\begin{algorithmic}[1]\small 
		\State Initialize the replay memory $R$, probability $\epsilon$, current network weights $\theta$, target network weights $\theta^{-}$. Set episode $=1$. \label{dqn_algo:hyp_initialize}
		\State Initialize the current network $Q(s,a|\theta)$ and the target network $Q(s,a|\theta^{-})$. \label{dqn_algo:network_initialize}
		\For{episode=1:E}
		\State initialize the environment by setting $t=1$. 
		    \For{t=1:T}
    		\State Select an action $ a $: \label{dqn_algo:action_select_start}
    		\State \quad select a random action $a \in \mathcal{A}$ with probability $ \epsilon $,
    		\State \quad otherwise select $ a = \max_{a} Q(x(t),a|\boldsymbol{\theta}) $. \label{dqn_algo:action_select_end}
    		\State Perform action $ a $. \label{dqn_algo:action_perform}
			\State Observe the reward $ r(t) $ and the new state $ x(t+1) $. \label{dqn_algo:obs_reward}
		    \State Store experience $ \left\{x(t),a(t),r(t),x(t+1)\right\} $ in replay \Statex \quad \quad \quad memory. \label{dqn_algo:store_exp}
		    \State Sample a random mini-batch $B$ of transitions \Statex \quad \quad \quad $x(t'),a(t'),r(t'),x(t'+1)$ from the replay memory. \label{dqn_algo:train_start}
		    \State Calculate target value:
		    \If{$x(t'+1)$ is a terminal state}
		    \State $y(t')=r(t')$,
		    \Else
		    \State $y(t')=r(t')+\max_{a}Q(x(t'+1,a|\theta^{-}))$. \label{dqn_algo:train_end}
		    \EndIf
		    \State Update $\theta$ using the weight update rule as per Eqn. \eqref{eqn:weight_update_current}. \label{dqn_algo:weight_update_current}
		    \State Update target network by setting $\theta^{-}=\theta$ every $O$ steps. \label{dqn_algo:weight_update_target}
		    \State Episode ends if $x(t+1)$ is the terminal state.
		    \EndFor
		\EndFor
		\end{algorithmic}
	\label{algorithm - dqn}
\end{algorithm}

%% file: 6-results.tex
\section{Results}
\label{sec:results}


In this section, we present the results in terms of average AoI at the TBS achieved for MAF-MAD and DQN schedulers, and compare them to three baseline schedulers (See subsection V.A). We first focus on smaller IoT network setting (with $10$'s IoT devices) and investigate the efficacy of proposed algorithms against the baseline ones. Furthermore, we also investigate the effect of various control parameters, such as 
the number of channels, IoT devices and UAVs. Finally, we discuss the comparative analysis of all the schedulers in the case of larger IoT network setting as well.

\subsection{Baseline schedulers}

\begin{itemize}
    \item \textbf{Maximal AoI First (MAF) scheduler} -- devices with highest AoI at both the UAV and TBS are selected for sampling and updating respectively \cite{sun2018age}. For sampling, each UAV selects $L$ devices with highest $AoI_m^{UAV}(t)$ out of the $m_n$ devices associated to it. Similarly for updating, $K$ devices out of the total $M$ with highest $AoI_m^{TBS} (t)$ get to update their packets.
    
    \item \textbf{Round Robin (RR) Scheduler} -- under RR, the available channels are assigned in equal and circular fashion among the devices which ensures fairness among the devices \cite{muller2018evaluation}
    
    \item \textbf{Random scheduler} -- This a baseline scheduler where a random set of devices are selected to be sampled and updated under the pre-specified channel constraints.

\end{itemize}

\subsection{Simulation setting}

The IoT devices and UAVs are placed randomly \cite{lyu2016placement} in an area of $l=b=$1000m and observation interval $T$=10 slots. Simulations for the DQN scheduler are performed using the tf-agents library \cite{GitHubte58:online}. The fully connected neural net has two hidden layers and the architecture is $|\mathcal{X}|\times{H_1}\times{H_2}\times |\mathcal{A}|$. The DQN is implemented on an NVIDIA DGX station with an Intel Xeon E5-2698 v4 CPU and an NVIDIA Tesla V100 GPU (32 GB memory). Data communication between CPU and GPU takes place through a PCIe 3.0 X16 slot. DQN hyper-parameters are shown in in Table \ref{tab:dqn_params}. The results for the other schedulers are averaged over $10^4$ runs. 


\begin{table}[htb]
 \centering
 \small
 \vspace{-0.1in}
 \caption{DQN Parameters}
    \begin{tabular}{|c|c|}
    \hline
    Neurons in hidden layer 1 $H_1$   & 1024 \\ \hline
    Neurons in hidden layer 2 $H_1$   & 1024 \\ \hline
    Relay memory size $R$ & 50,000 \\ \hline
    Mini batch size $B$   & 16   \\ \hline    
    Learning rate $\beta$   & 0.001    \\ \hline
    Learning rate decay rate & 0.95 \\ \hline
    Learning rate decay step & 10,000 \\ \hline
    Discount factor $\gamma$   & 1    \\ \hline    
    Activation function & ReLU \\ \hline
    Update steps $O$ & 10 \\ \hline
    Optimizer & Adam \\ \hline
    Total episodes $E$ & $10^{6}$\\ \hline

 \end{tabular}
 \label{tab:dqn_params}
\end{table}

\subsection{Smaller IoT networks}

These IoT networks contain 5-10 IoT devices. Such networks are used for monitoring indoors/outdoors  particulate matter pollutants where a single type of sensor is able to monitor particulates of various sizes \cite {reddy2020improving}. We use $L=K=1$ and simulate 2 scenarios: first with $M=12$ IoT devices equally assigned to $N=2$ UAVs, and then $M=9$ IoT devices equally assigned to $N=3$ UAVs.

\subsubsection{Under Ideal Conditions}


Here the devices generate traffic using a generate-at-will model with no packet loss in any of the channels. The results shown in Fig. \ref{fig:ideal_age} confirm that MAF-MAD is the optimal scheduling policy and outperforms all other scheduling policies. As DQN is able to learn the the optimal scheduling, it's performance also converges to MAF-MAD's performance. MAF and RR have similar performance while random performs the worst. With MAF-MAD and DQN performing similarly, their performance exceeds MAF, RR and random by 5\%-10\%, 5\%-10\% and 11\%-17\% respectively. 


    \vspace{-0.05 in}


\subsubsection{Under general conditions}


\begin{figure}[H]
\vspace{-.05in}
\subfigure[\label{fig:ideal_age}]{
\epsfig{figure=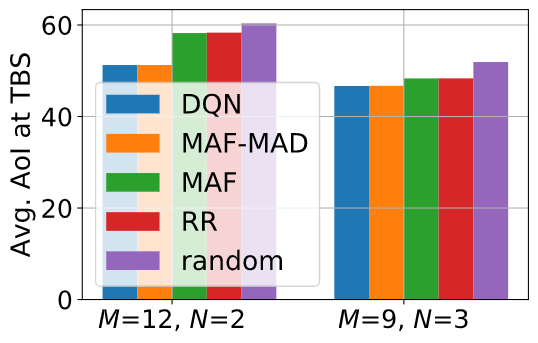,width=1.65 in, keepaspectratio}}
\subfigure[\label{fig:non_ideal_age}]{
\epsfig{figure=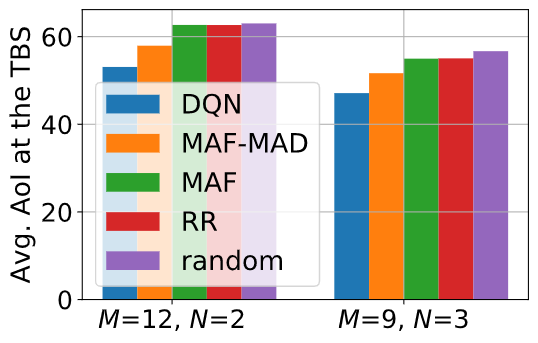,width=1.65 in, keepaspectratio}}
\caption{$AoI^{TBS}(T)$ under (a) ideal conditions (b) general conditions for smaller IoT networks}
\vspace{-.05in}
\end{figure}

The results are shown in Fig. \ref{fig:non_ideal_age} and we see that DQN outperforms MAF-MAD and all other schedulers under these conditions. Therefore, MAF-MAD scheduler is no longer optimal under these conditions. MAF-MAD is followed by MAF, RR and random. DQN's performance exceeds MAF-MAD, MAF, RR and random by 9\%-10\%, 16\%-18\%, 17\%-18\% and 18\%-20\% respectively. 



\begin{figure}[H]
\vspace{-.05in}
\subfigure[\label{fig:individual_AoI_loss}]{
\epsfig{figure=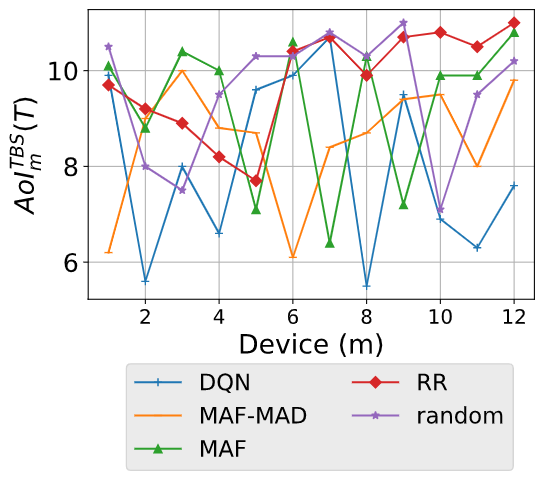,width=1.65 in, keepaspectratio}}
\subfigure[\label{fig:individual_AoI_period}]{
\epsfig{figure=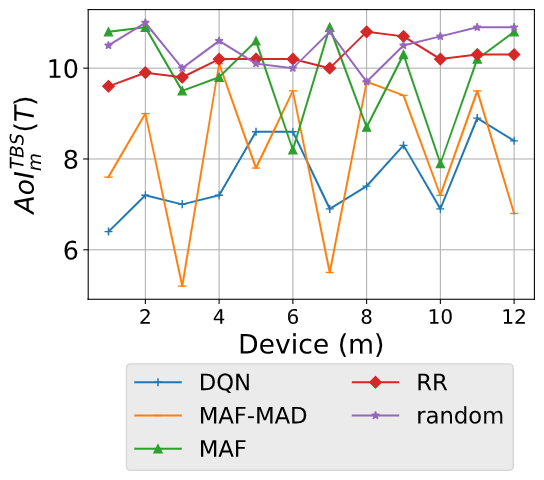,width=1.65 in, keepaspectratio}}
\caption{Individual AoI under only (a) packet loss (b) periodic traffic generation}
\vspace{-.05in}
\end{figure}

To understand why DQN has the best performance, we show the individual AoI of the devices at TBS at $t=T$ for lossy channel with generate-at-will traffic generation in Fig. \ref{fig:individual_AoI_loss} and non-lossy channels with periodic traffic generation in Fig. \ref{fig:individual_AoI_period}. $N=2$ UAVs are providing coverage to $M=12$ IoT devices with the IoT devices equally distributed between the UAVs. For Fig. \ref{fig:individual_AoI_loss}, 
while the results could be shown for any channel condition, we pre-assign the packet loss probabilities $l_{s,m}$ and $l_{u,m}$ for ease of reproducibility \cite{8885156}. Under lossy channel conditions, focusing on channels that improve packet receptions will minimize average AoI and the DQN based scheduler will learn this. The average AoI per device 
for DQN, MAF-MAD, MAF, RR and random were - \textbf{8.00}, 8.55, 9.29, 9.34, 9.58 respectively and as expected, DQN has the best performance. Similarly in Fig. \ref{fig:individual_AoI_period}, $p_m$ is selected randomly between 2,3,4 and pre-assigned for the devices. 
The average AoI for DQN, MAF-MAD, MAF, RR and random were - \textbf{7.65}, 8.10, 9.88, 10.18, 10.47 respectively, and it can be seen that the DQN based scheduler has lowest AoI for most of the devices. This means DQN is able to learn the periods of packet generation at the devices and sample them accordingly, which the other schedulers are unable to do. 


We now investigate the effect of different parameters on average  AoI at the TBS under general conditions.

\paragraph{Variable $M$ (IoT devices)}
For $L$=2, $K$=1, we deploy $N=3$ UAVs and vary the total number of IoT devices $M$. The IoT devices are assigned equally among the 3 UAVs.
The results are shown in Fig. \ref{fig:variable_M} - as $M$ increases, average  AoI at the TBS increases too. This is because while $M$ increases, number of channels ($L$, $K$) remain same due to which each device now has lesser chances of getting sampled and updated, resulting in increased average  AoI at the TBS. 

\paragraph{Variable $N$ (UAVs)}
With $L$=2, $K$=1, we place $M=10$ IoT devices and vary the number of UAVs deployed. The IoT devices are assigned equally among the UAVs with the remaining IoT devices assigned to the last UAV, if any.
The results are shown in Fig. \ref{fig:variable_N} - for constant $M$, $K$ and $L$, average  AoI at the TBS remains mostly same.
This is because even though $N$ is increased, $K$ remains the same for all. Hence capacity of information transfer from the UAVs to the TBS is the same which results in similar average  AoI at the TBS. 




\begin{figure}[htb]
\vspace{0.05in}
\subfigure[\label{fig:variable_M}]{
\epsfig{figure=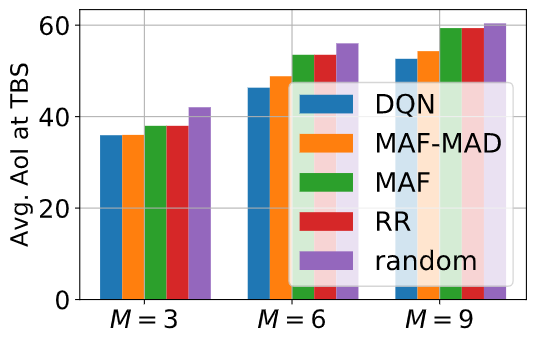,width=1.65 in, keepaspectratio}}
\subfigure[\label{fig:variable_N}]{
\epsfig{figure=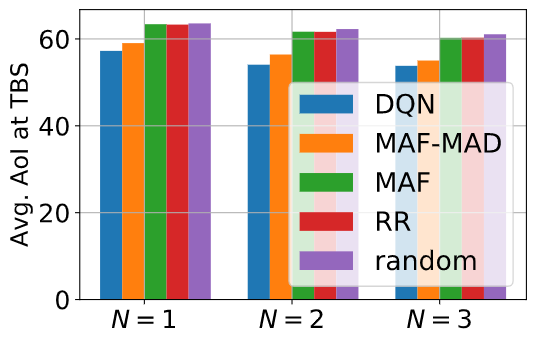,width=1.65 in, keepaspectratio}}
\caption{$AoI^{TBS}(T)$ for varying (a) $M$ (b) $N$. }
\vspace{-0.05in}
\end{figure}




\paragraph{Varying number of channels ($L$ and $K$)}

First for fixed $K$=1, $M$=9 IoT devices are allocated equally between $N$=3 UAVs. $L$ is varied and the results for average  AoI at the TBS are shown in Fig. \ref{fig:variable_L}. As $L$ increases, we don't see any noticeable improvement in average  AoI at the TBS. A higher $L$ will improve the ${AoI}^{UAV}(T)$ (see Fig. \ref{fig:variable_L_UAV}), but this benefit is not passed on to the 
TBS as seen from Fig. \ref{fig:variable_L}. This indicates that increasing $L$ doesn't help much in improving average  AoI at the TBS. The average AoI at the UAV for $L$=3 in Fig. \ref{fig:variable_L_UAV} is the same for all the schedulers as $L=m_n=3, \hspace{.06in} \forall n \in N$, due to which all the devices get equally sampled.

Then for the same setting as above, we vary $K$ keeping other parameters constant. $L$ is set to 3 so that each UAV can sample all 3 devices under it at each slot. Therefore the ${AoI}^{UAV}(T)$ is the same for all the four schedulers as seen from Fig. \ref{fig:variable_K_UAV}. Then average  AoI at the TBS is shown in Fig. \ref{fig:variable_K}. Here we see an improvement as $K$ is increased due to which more information can be transferred from the UAVs to the TBS at each slot resulting in lower average  AoI at the TBS. This shows that the channels between the UAVs and the TBS is more important and critical bottleneck compared to the channels between the IoT devices and the UAVs in improving average  AoI at the TBS.





\begin{figure}[htb]
\vspace{0.05in}
\subfigure[\label{fig:variable_L_UAV}]{
\epsfig{figure=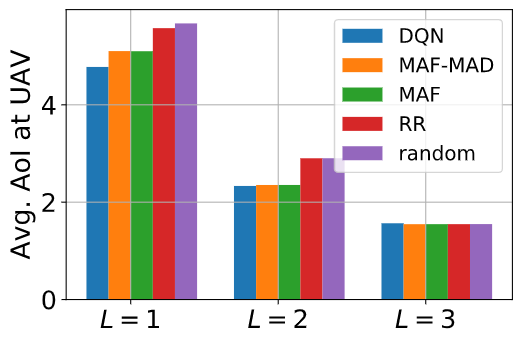,width=1.65 in, keepaspectratio}}
\subfigure[\label{fig:variable_L}]{
\epsfig{figure=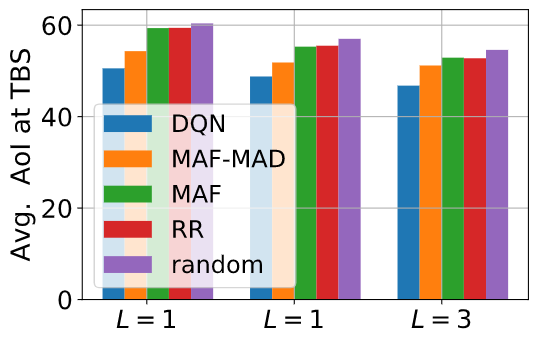,width=1.65 in, keepaspectratio}}
\caption{(a) $AoI^{UAV}(T)$ (b)$AoI^{TBS}(T)$ for varying $L$. }
\vspace{-0.05in}
\end{figure}





\begin{figure}[htb]
\vspace{0.05in}
\subfigure[\label{fig:variable_K_UAV}]{
\epsfig{figure=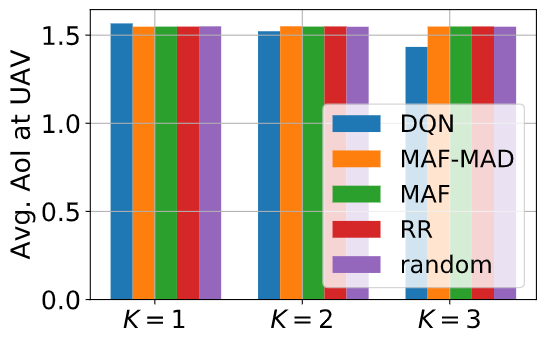,width=1.65 in, keepaspectratio}}
\subfigure[\label{fig:variable_K}]{
\epsfig{figure=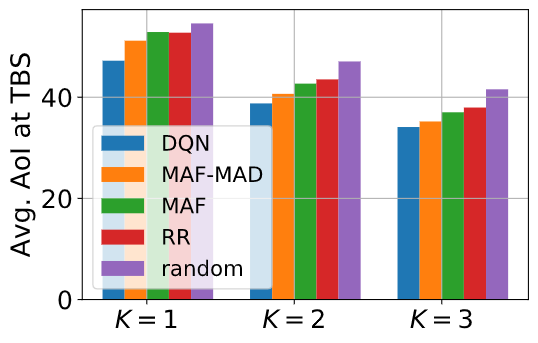,width=1.65 in, keepaspectratio}}
\caption{(a) $AoI^{UAV}(T)$ (b)$AoI^{TBS}(T)$ for varying $K$. }
\vspace{-0.05in}
\end{figure}


\subsection{Larger IoT networks}


Now we consider networks requiring around 50 IoT devices and these types of networks are deployed where low number of devices are not able to capture enough information. E.g. in agricultural fields, multiple sensors for capturing moisture, temperature levels etc are needed \cite{elijah2018overview}.

\subsubsection{Experimental Results}





\begin{figure}[htb]
\vspace{0.05in}
\subfigure[\label{fig:dqn_bad_ideal}]{
\epsfig{figure=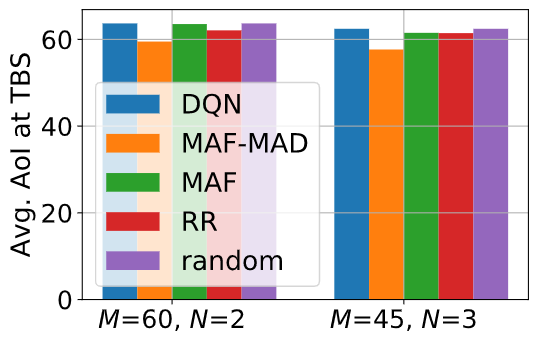,width=1.65 in, keepaspectratio}}
\subfigure[\label{fig:dqn_bad_non_ideal}]{
\epsfig{figure=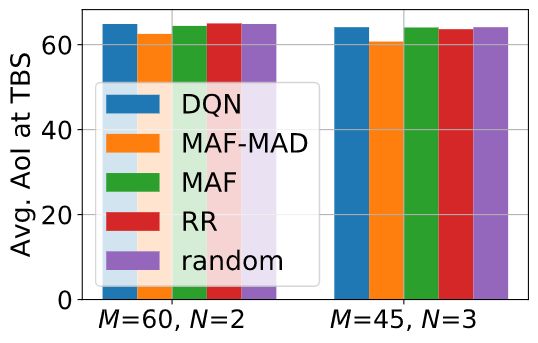,width=1.65 in, keepaspectratio}}
\caption{$AoI^{TBS}(T)$ under (a) ideal conditions (b) general conditions for larger IoT networks.}
\vspace{-0.05in}
\end{figure}

We set $L=K=2$ and simulate two scenarios: first we place $M=60$ IoT devices and assign then equally among $N=2$ UAVs, and secondly $M=45$ devices are assigned equally among $N=3$ UAVs. The results under ideal and general scenarios are shown in Fig. \ref{fig:dqn_bad_ideal} and Fig. \ref{fig:dqn_bad_non_ideal} respectively.


This result helps us draw key challenges and a future road-map for further research. From Fig. \ref{fig:dqn_bad_ideal} and Fig. \ref{fig:dqn_bad_non_ideal}, it is evident that DQN does badly in terms of AoI minimization for larger networks with MAF-MAD having the best performance under both conditions. This is because, as the network size increases, the action space also increases and DQN \textit{doesn't} provide satisfactory performance when the action space is very large \cite{zahavy2018learn}. Due to this, works involving DQN for scheduling in IoT networks to improve AoI have been limited to smaller networks \cite{9014214, zhou2019deep, hu20201, abd2019deep, song2020optimal}. An alternative approach for larger scenarios is distributed multi-agent learning \cite{xu2015distributed} where each UAV and the TBS acts as independent agents, and this will be explored in a future version of this work. 


%% file: 7-conclusion.tex
\section{Conclusion}
\label{sec:conclusion}

We prove and show the optimality of a MAF-MAD scheduler for improving AoI in a UAV-relayed IoT network with two hops, where the UAVs act as static ABS that relay packets between the IoT devices and the TBS under no packet loss in any of the channels and generate-at-will traffic generation model at the IoT devices. However when packet loss in the channels and periodicity in packet generation by the IoT devices are considered, MAF-MAD scheduler is no longer optimal. Therefore we propose a DQN based scheduler which performs better than the MAF-MAD scheduler under these conditions when the network size is small. However the proposed DQN scheduler doesn't scale well with the network size due to the inherent limitations of DQN. We also show that the channels between the UAVs and the TBS play a much important role in improving AoI at the TBS compared to channels between the IoT devices and UAV, which points to the importance of backhaul links compared to access links for AoI minimization. 
In our future work, we will consider mobility and optimal UAV placement under UAV-relayed IoT networks with the objective of improving information freshness. Additionally, to account for larger networks, distributed learning approaches will also be investigated.